\documentclass[journal]{IEEEtran}
\usepackage[utf8]{inputenc}
\usepackage{bm}
\usepackage{amsmath,epsfig,amssymb,dsfont,amsthm}
\usepackage{hyperref}
\usepackage{verbatim}
\usepackage{caption}
\usepackage{subcaption}
\usepackage[dvipsnames]{xcolor}
\usepackage{graphicx}
	\graphicspath{{./figures/}}
\usepackage{float}
\usepackage{url}
\usepackage{eurosym}
\usepackage{setspace}
\usepackage{balance}
\usepackage{dsfont}
\usepackage{bbm}
\usepackage{array}
\usepackage{cite}
\usepackage[ruled,vlined]{algorithm2e}
\usepackage{thmtools}
\usepackage{comment}
\usepackage[normalem]{ulem}
\usepackage{mathtools}
\usepackage{framed}

\usepackage{tabularx,booktabs}
\newcolumntype{A}{<{\raggedright\arraybackslash}X}
\newcolumntype{B}{>{\raggedleft\arraybackslash}{\hsize=.3\hsize}X}
\newcolumntype{C}{>{\hsize=.24\hsize}X}

\newcommand{\blambda}{\boldsymbol{\Lambda}}
\newcommand{\bA}{\boldsymbol{A}}

\newcommand{\bL}{\boldsymbol{\mathcal L}}
\newcommand{\cL}{\mathcal L}

\newcommand{\R}{\mathcal R}

\newcommand{\bT}{\mathsf{T}}
\newcommand{\bE}{\mathbb{E}}

\newcommand{\bmu}{\boldsymbol \mu}
\newcommand{\bpsi}{\boldsymbol \psi}
\newcommand{\bzeta}{\boldsymbol \zeta}
\newcommand{\btheta}{\boldsymbol \theta}

\theoremstyle{plain}
\newtheorem{Thm}{Theorem}

\newtheorem{Asm}{Assumption}

\theoremstyle{remark}

\newcommand{\qedsymb}{\hfill\ensuremath{\blacksquare}}                 
\allowdisplaybreaks
\title{Detection of Malicious Agents in Social Learning} 
\author{\IEEEauthorblockN{Valentina Shumovskaia, Mert Kayaalp, and Ali H. Sayed}\\
$\newline$
\IEEEauthorblockA{
École Polytechnique Fédérale de Lausanne (EPFL)
}
\thanks{Emails: $\{$valentina.shumovskaia, mert.kayaalp, ali.sayed$\}$@epfl.ch.
}
}

\IEEEoverridecommandlockouts
\begin{document}

\maketitle

\begin{abstract}
    \color{black}{Non-Bayesian social learning} is a framework for distributed hypothesis testing aimed at learning the true state of the environment. 
    Traditionally, the agents are assumed to receive observations conditioned on the same true state, although it is also possible to examine the case of heterogeneous models across the graph. 
    One important special case is when heterogeneity is caused by the presence of malicious agents whose goal is to move the agents toward a wrong hypothesis.
    In this work, we propose an algorithm that allows discovering the true state of every individual agent based on the \textit{sequence} of their beliefs. In so doing, the methodology is also able to locate malicious behavior.
\end{abstract}


\section{Introduction and Related Work}
    {\color{black}Non-Bayesian} social learning algorithms~\cite{jadbabaie2012non, zhao2012learning, salami2017social, nedic2017fast, molavi2017foundations, molavi2018theory, bordignon2020adaptive, bordignon2022partial, lalitha2018social, inan2022social, 9132712, 9670665} solve the distributed hypothesis problem in a {\color{black}\emph{locally} Bayesian fashion}. 
    These algorithms learn the underlying true state of nature by observing streaming data arriving at the agents and  conditioned on that state. 
    The key difference with Bayesian solutions~\cite{gale2003bayesian, acemoglu2011bayesian, hkazla2021bayesian} is that {\color{black}non-Bayesian} social learning does not require each node to know the full graph topology or likelihood models used by every other node. 
    These features enable fully decentralized implementations. 
    {\color{black}Social learning frameworks can be applied in many contexts, including in sensor network detection~\cite{rabbat2004decentralized, rabbat2005robust}, distributed machine learning~\cite{hu2023non, bordignon2022partial}, and the modeling of user opinions on social graphs~\cite{shumovskaia2023discovering}.}
    
    Under social learning, agents update their beliefs (or confidences) on each possible hypothesis, ensuring that the total confidence adds up to $1$. 
    At every time instant, each agent receives an observation conditioned on the state of the environment and uses its local likelihood models to perform a local Bayesian update starting from its current belief vector.
    {\color{black}This step is followed by a communication stage where agents exchange and fuse beliefs with neighbors.
    These steps are repeated until convergence.} 
    
    Many existing works on social learning assume that the observations received by each agent arise from \textit{one} true state of the environment. 
    {\color{black}Others study nonhomogeneous models, such as~\cite{shumovskaia2023community}, which focuses on community networks where each community has its own truth.}
    {\color{black}The main conclusion is that if the malicious agents are sparsely located in the network, it often becomes impossible to track such agents based just on their belief. Also, additional defense strategies against malicious agents can be implemented~\cite{mitra2019new, su2019defending}.}
    
    In this work, we develop a centralised algorithm for identifying the true state associated with each agent, even when the final belief of an agent may be pointing toward another conclusion due to the interactions over the graph. In this way, the method is able to identify malicious agents as well. 
    There is no question that this is an important issue that deserves attention~\cite{zhang2019measuring, 9229115, krishnamurthy2012afriat, 8013830, illiano2015detecting, pang2018towards, zhang2016misinformation, smith2021automatic, egele2013compa, tomaiuolo2020survey, fornacciari2018holistic, sadiq2021aggression, 8846206, 9023355}. 
    For instance, over social networks, it is critical to identify  users that have unwarranted intentions and aim to force the network to reach erroneous conclusions~\cite{zhang2016misinformation, smith2021automatic, egele2013compa}, as well as to discover trolls~\cite{tomaiuolo2020survey, fornacciari2018holistic, sadiq2021aggression} and measure their impact on performance~\cite{zhang2019measuring}. 
    The same techniques can be used to locate malfunctioning agents~\cite{krishnamurthy2012afriat}.

    {\color{black}There are other works that deal with similar objectives, albeit under different assumptions and considerations. For example, the works~\cite{vempaty2013distributed, chen2008robust} address Byzantine agent detection but assume a collection of  i.i.d. data conditioned on each agent's true state. In comparison, our approach collects correlated shared beliefs from inter-agent communication. Other methods leverage temporal and spatial correlations~\cite{shahid2012quarter, lai2022identifying, rezvani2013robust} and topological features~\cite{gu2020malicious}, but they lack theoretical guarantees. Our method's advantage is its formulation as an inverse modeling problem, ensuring convergence based on a suitable choice of the step-size parameter. Additionally, there are fully distributed approaches for malicious agent detection based on consensus constructions~\cite{gu2020malicious}, where agents store their neighbors' signal history and exclude suspicious nodes from communication. In social learning, a similar algorithm~\cite{10124228} adapts the initial graph topology based on each agent's detected true state, involving additional computational efforts. In comparison, our method maintains the original topology, preserving the network structure while effectively identifying malicious agents without altering it.}

\section{Social Learning Model}\label{sec:model}
    A set of agents $\mathcal{N}$ builds confidences on each hypothesis $\theta$ from a finite set $\Theta$ through interactions with the environment and among the agents. 
    The agents communicate according to a fixed combination matrix $A \in \mathcal [0,1]^{\mathcal N \times \mathcal N}$, where each nonzero element $a_{\ell, k} > 0$ indicates a directed edge from agent $\ell$ to agent $k$ and defines the level of trust that agent $k$ gives to information arriving from agent $\ell$. 
    Each agent $k$ assigns a total confidence level of $1$ to its neighbors. This assumption makes the combination matrix $A$ left stochastic, i.e., 
    \begin{align}
        \sum_{\ell \in \mathcal N} a_{\ell k} = 1,\;\forall k \in \mathcal N
    \end{align}
    Another common assumption, ensuring global truth learning for homogeneous environments, is that $A$ is strongly connected. This implies the existence of at least one self-loop with a positive weight and a path with positive weights between any two nodes~\cite{Sayed_2014}.
    This condition allows us to apply the Perron-Frobenius theorem~\cite[Chapter 8]{horn2009},~\cite{sayed_2023}, which ensures that the power matrix $A^s$ converges exponentially to $u\mathds 1^\bT$ as $s \rightarrow \infty$. Here, $\mathds{1}$ is the vector of all 1s and $u$ is the Perron eigenvector of $A$ associated with the eigenvalue at $1$ and is normalized as follows:
    \begin{align}
        A u = u,\qquad u_\ell > 0,\qquad \textstyle\sum_{\ell \in \mathcal N} u_\ell = 1.
    \end{align}

    Each agent assigns an initial \textit{private} belief $\bmu_{k,0}(\theta)\in[0,1]$ to each hypothesis $\theta\in\Theta$, forming a probability mass function with the total confidence summing up to $1$, i.e., $\sum_{\theta}\bmu_{k,0}(\theta) = 1$.
    {\color{black}To avoid excluding any hypothesis initially, we assume $\bmu_{k,0}(\theta) > 0$ for all $\theta$.}
    Subsequently, agents iteratively update their belief vectors by interacting both with the environment and with their neighbors.
    At each time instance $i$, agent $k$ receives an observation from the environment conditioned on its true state, denoted by $\bzeta_{k,i} \sim L_k(\zeta|\theta_k^\star)$ or $L_k(\theta_k^\star)$ for brevity.
    In this notation, the observation $\bzeta_{k,i}$ arises from the likelihood model $L_k(\zeta|\theta_k^{\star})$, which is parameterized by the unknown model $\theta_k^{\star}$. 
    For example, the entire network may be following the same and unique model $\theta^\star$, while a few malicious agents may be following some other model $\theta\neq\theta^\star$.
    The observations $\{\bzeta_{k,i}\}$ are assumed to be  independent and identically distributed (i.i.d.) over time. 
    The local Bayesian update performed by agent $k$ at time $i$ takes the following form~\cite{bordignon2020adaptive}:
    \begin{align}
        \bpsi_{k,i}(\theta) = \frac{L_k^\delta(\bzeta_{k,i}\mid\theta)\bmu^{1-\delta}_{k,i-1}(\theta)}{\sum_{\theta'\in\Theta}L^\delta_k(\bzeta_{k,i}\mid\theta')\bmu^{1-\delta}_{k,i-1}(\theta')},\quad \forall k\in\mathcal{N}, \label{eq:adapt_adaptive}
    \end{align}
    where $\delta \in (0,1)$ plays the role of an adaptation parameter and it controls the importance of the newly received observation relative to the information learned from past interactions. 
    The denominator in~(\ref{eq:adapt_adaptive}) serves as a normalization factor, ensuring that the resulting $\bpsi_{k,i}$ is a probability mass function. 
    We refer to $\bpsi_{k,i}$ as the \textit{public} (or intermediate) belief due to the next communication step, which involves a geometric averaging computation~\cite{nedic2017fast, lalitha2018social, zhao2012learning}:
    \begin{align}
        &\bmu_{k,i}(\theta)=\frac{\prod_{\ell\in\mathcal{N}_k}\bpsi^{a_{\ell k}}_{\ell,i}(\theta)}{\sum_{\theta'\in\Theta}\prod_{\ell\in\mathcal{N}_k}\bpsi^{a_{\ell k}}_{\ell,i}(\theta')}, \quad \forall k\in\mathcal{N}. \label{eq:combine}
    \end{align}

    At each iteration $i$, each agent $k$ estimates its true state $\theta_k^\star$ based on the belief vector (either private or public) by selecting the hypothesis with the highest confidence:
    \begin{align}
        \widehat{\btheta}_{k,i} \triangleq \arg\max_{\theta\in\Theta}\bmu_{k,i}(\theta).
        \label{eq:truestate_est0}
    \end{align}
    In the homogeneous environment case~\cite{nedic2017fast, lalitha2018social, zhao2012learning, bordignon2020adaptive}, i.e., when $\theta_k^\star=\theta^\star$ for each $k$, it can be proved that every agent finds the truth asymptotically with probability $1$. 
    
    The work~\cite{shumovskaia2023community} considers nonhomogeneous environments with community-structured graphs; it establishes that, as $\delta \to 0$, the entire network converges to {\textit{one}} solution, while in contrast, a larger $\delta$ activates the mechanism of \textit{local} adaptivity.
    While this property works well with community-structured graphs, some \textit{sparsely} located malicious agents might be heavily influenced by their neighbors or require too large $\delta$. The method we derive estimates the true state of each agent in an inverse manner, allowing it to operate effectively with graphs of general structure and with any $\delta$.

\section{Inverse Modeling}\label{sec:alg}
    In this section, we explain how we can identify malicious agents (or the true state $\theta_k^\star$ for each agent) by observing sequences of public beliefs. 
    Importantly, we will not assume knowledge of the combination matrix $A$.
    
    To begin with, we introduce the following common assumption, essentially requiring the observations to share the same support region~\cite{bordignon2022partial, shumovskaia2022explainability, shumovskaia2023discovering}.
    \begin{Asm}[\bf{Bounded likelihoods}]
        \label{asm:support}
        There exists a finite constant $b > 0$ such that for all $k \in \mathcal N$:
        \begin{align}
            \Bigg|\log \frac {L_k(\boldsymbol\zeta \mid \theta)}{L_k(\boldsymbol\zeta \mid \theta')} \Bigg| \leq b
        \end{align}
        for all $\theta,\;\theta' \in \Theta$ and $\boldsymbol\zeta$.
        \qedsymb
    \end{Asm}

    \noindent Now, consider a \textit{sequence} of public beliefs measured closer to the steady state:
    \begin{align}
        \{\bpsi_{k,i}\}_{i\gg 1},\;k\in\mathcal N
        \label{eq:problem}
    \end{align}
    
    \noindent When an agent cannot distinguish between $\theta^\star_k$ and another $\theta$ due to $L_k(\theta_k^\star) = L_k(\theta)$, we will treat this $\theta$ as a valid model for the agent as well.
    To accommodate this possibility, we define $\Theta_k^{\star}$ as the optimal hypotheses subset for each individual agent, denoted by $\Theta_k^\star = \{\theta_k^\star\} \cup \{\theta\neq\theta_k^\star \mid L_k(\theta) = L_k(\theta_k^\star)\}$.
    Then, we reformulate the problem by stating that our aim is to recover the optimal hypotheses subset for each agent:
    \begin{align}
        \{\Theta_k^\star\}, \;k\in\mathcal N.
    \end{align}

    We denote the level of informativeness of any pair of hypotheses $\theta,\theta'\in\Theta$ at each agent $k$ by:
    \begin{align}
        d_k(\theta,\theta') \triangleq \mathbb E_{\bzeta_k \sim L_k(\theta_k^\star)} \log \frac{L_k(\bzeta_k|\theta)}{L_k(\bzeta_k|\theta')}
        \label{eq:Ldist_estimation}
    \end{align}
    It is clear that this value is equal to zero if both $\theta$ and $\theta'$ belong to the optimal subset $\Theta_k^\star$. Additionally, $d_k(\theta_k^\star, \theta)$ will be positive for any $\theta\notin \Theta^\star_k$ since
    \begin{align}\label{eq:prop1}
        d_k(\theta_k^\star,\theta) = D_{\textup{KL}}\left(L_k\left(\theta_k^\star\right) \mid\mid L_k\left(\theta\right)\right) > 0 
    \end{align}
    and, in turn, $d_k(\theta, \theta_k^\star)$ is always negative:
    \begin{align}\label{eq:prop2}
        d_k(\theta,\theta_k^\star) = - D_{\textup{KL}}\left(L_k\left(\theta_k^\star\right) \mid\mid L_k\left(\theta\right)\right) < 0
    \end{align}
    Here, $D_{\textup{KL}}$ denotes the Kullback-Leibler divergence between two distributions:
    \begin{align}
        D_{\textrm{KL}} \big(L_k(\theta^\star) \mid\mid L_k(\theta)\big) \triangleq \bE_{\bzeta \sim L_k(\bzeta \mid \theta^\star)} \log \frac{L_k(\bzeta\mid \theta^\star)}{L_k(\bzeta\mid\theta)}
    \end{align}
    Properties (\ref{eq:prop1})--(\ref{eq:prop2}) allow us to conclude that the optimal hypotheses subset $\Theta_k^\star$ consists of all $\theta$ for which:
    \begin{align}
        \Theta_k^\star = \{\theta \colon d_k(\theta, \theta') \geq 0,\;\ \forall \theta'\in\Theta \}
        \label{eq:Theta_star_redef}
    \end{align}
    Our aim is to develop an algorithm that learns $\Theta_k^\star$ based on the available information~(\ref{eq:problem}).

    In~\cite[Appendix A]{shumovskaia2022explainability}, it is shown that the adaptive social learning iterations (\ref{eq:adapt_adaptive})--(\ref{eq:combine}) can be expressed in the following compact linear form:
    \begin{align}
        \blambda_i = (1-\delta)A^\bT \blambda_{i-1} + \delta \bL_i
        \label{eq:recursion}
    \end{align}
    where $\blambda_i$ and $\bL_i$ are matrices of size $|\mathcal N| \times (|\Theta| - 1)$, and for each $k$ and $j$, their entries take the log-ratio form:
    \begin{align}
        &[\boldsymbol{\Lambda}_{i}]_{k,j} \triangleq\log\frac{\boldsymbol{\psi}_{k,i}(\theta_0)}{\boldsymbol{\psi}_{k,i}(\theta_j)},
        \; \;[\boldsymbol{\mathcal{L}}_{i}]_{k,j}\triangleq\log\frac{L_k(\boldsymbol{\zeta}_{k,i}\mid\theta_0)}{L_k(\boldsymbol{\zeta}_{k,i}\mid\theta_j)}.
        \label{eq:loglikelihood}
    \end{align}
    for any ordering $\Theta = \{\theta_0,\dots,\theta_{|\Theta|-1}\}$. The expectation of $\bL_i$, relative to the observations $\{\bzeta_{k,i}\}_k$, is given by:
    \allowdisplaybreaks[0]
    \begin{align}
        [\overline \cL]_{k,j} \triangleq [\bE \bL_i]_{k,j} =\textrm{ }&  D_{\textup{KL}}\left(L_k\left(\theta_k^\star\right) \mid\mid L_k\left(\theta_j\right)\right) \nonumber\\
        &- D_{\textup{KL}}\left(L_k\left(\theta_k^\star\right) \mid\mid L_k\left(\theta_0\right)\right),
        \label{eq:L_exp}
    \end{align}
    \allowdisplaybreaks
    and it allows us to rewrite~(\ref{eq:Ldist_estimation}) in a slightly different manner:
    \allowdisplaybreaks[0]
    \begin{align}
        d_k(\theta_{j_1},\theta_{j_2}) = [\overline \cL]_{k,j_2} - [\overline \cL]_{k,j_1}
        \label{eq:Ldist_estimation3}
    \end{align}
    \allowdisplaybreaks
    Furthermore, it is shown in  \cite{shumovskaia2023discovering} that we can estimate $\overline \cL$ by utilizing the publicly exchanged beliefs with the following accuracy~\cite[Theorem 2]{shumovskaia2023discovering}:
    \begin{align}
        &\;\limsup_{i\rightarrow\infty} \bE \|\widehat \bL_{i} - \overline \cL\|_{\textrm F}^2 \nonumber\\
        &\;\leq \frac 1M \textup{Tr} \left(\R_{\bL}\right) + O(\mu/\delta^2) + O\left( 1 / \delta^5 M^2 \right)
        \label{eq:prev_result}
    \end{align}
    where $\mu$ is a small positive learning rate for a stochastic gradient implementation, $M$ is a batch size of data used to compute the estimate $\widehat{\bL}_i$, and $R_{\bL} \triangleq \bE \left(\bL_i - \overline \cL \right) \left(\bL_i - \overline \cL \right)^\bT$.
    Thus, the informativeness~(\ref{eq:Ldist_estimation3}) can be estimated by using
    \begin{align}
        \widehat {\boldsymbol d}_k(\theta_{j_1},\theta_{j_2}) = [\widehat \bL]_{k,j_2} - [\widehat \bL]_{k,j_1}
        \label{eq:Ldist_estimation2}
    \end{align}
    where $\widehat \bL$ is the estimate of $\overline \cL$ from the last available iteration.
    Based on (\ref{eq:Theta_star_redef}), we can now identify the optimal hypotheses subset $\Theta_k^\star$ defined in~(\ref{eq:Theta_star_redef}) as follows:
    \begin{align}
        \widehat {\boldsymbol\Theta}_k \triangleq \arg\max_{\theta_{j_1}} \sum_{\theta_{j_2}} \mathbb I \left\{\widehat {\boldsymbol{d}}_k(\theta_{j_1}, \theta_{j_2}) > 0\right\}
    \end{align}
    where $\mathbb I \left\{ x \right\}$ is an indicator function that assumes the value $1$ when its argument is true and is $0$ otherwise. 

    We list the procedure in Algorithm~\ref{alg}, including the part related to estimating~(\ref{eq:prev_result}) by using \cite[Algorithm 1]{shumovskaia2023discovering}.

    \begin{algorithm}
        \KwData{
            At each time $i$: 
                $\left\{ \bpsi_{k,i}(\theta)\right\}_{k\in \mathcal N}$, $\delta$
        }
        \KwResult{Estimated combination matrix $\bA$; \\
        $\;\;\;\;\;\;\;\;\;\;\;\;$Estimated expected log-likelihood ratios 
        $\widehat{\bL}$;\\
        $\;\;\;\;\;\;\;\;\;\;\;$ Estimated set of true states for each agent,  $\widehat {\boldsymbol\Theta}_k$.\\
        }
        initialize $\bA_0$, $\widehat \bL_0$\\
        \Repeat{sufficient convergence}{
            Compute matrices $\blambda_i$:\\
            \For{$k\in\mathcal N$, $j=1,\dots,|\Theta|$}{
                \begin{flalign*}
                    &[\boldsymbol{\Lambda}_{i}]_{k,j} = \log\left(\boldsymbol{\psi}_{k,i}(\theta_0) / \boldsymbol{\psi}_{k,i}(\theta_j) \right)&&
                \end{flalign*}
            }
            \noindent Combination matrix update~\cite{shumovskaia2023discovering}:
            \begin{flalign*}
                &\boldsymbol{A}_i =  \boldsymbol{A}_{i-1} + \mu(1-\delta)\left(\blambda_{i-1} - M^{-1} \textstyle\sum_{j=i-M}^{i-1} \blambda_{j-1}\right) \nonumber&&\\
                &\textrm{ }\textrm{ }\textrm{ }\textrm{ }\textrm{ }\textrm{ }\times\left(\blambda_i^{\mathsf{T}} -  (1-\delta)\blambda_{i-1}^{\mathsf{T}}\boldsymbol{A}_{i-1} - \delta \widehat \bL_{i-1}^\bT \right).&&
            \end{flalign*}
            \noindent Log-likelihoods matrix update:
            \begin{flalign*}
                &\widehat \bL_i = \delta^{-1} M^{-1} \textstyle\sum_{j=i-M+1}^{i}\left(\blambda_j - (1-\delta)\bA_{i}^\bT \blambda_{j-1}\right)&&
            \end{flalign*}
            \begin{flalign*}
                &i = i+1&&
            \end{flalign*}
        }
        \noindent $\;$ Informativeness estimate for all agents $k \in \mathcal N$ and pairs of hypotheses $\theta_{j_1},\;\theta_{j_2}\in\Theta$:
        \begin{flalign*}
            &\widehat {\boldsymbol d}_k(\theta_{j_1},\theta_{j_2}) =   [\widehat   \bL_i]_{k,j_2} - [\widehat \bL_i]_{k,j_1}&&
        \end{flalign*}
        \noindent Optimal hypotheses set estimate for all agents $k\in\mathcal N$:
        \begin{flalign*}
            &\widehat {\boldsymbol\Theta}_k \triangleq \arg\max_{\theta_{j_1}} \textstyle\sum_{\theta_{j_2}} \mathbb I \left\{\widehat {\boldsymbol{d}}_k(\theta_{j_1}, \theta_{j_2}) > 0\right\}&&
        \end{flalign*}
        \caption{Inverse learning of heterogeneous states}
        \label{alg}
    \end{algorithm}

    The following result establishes the probability of error.
    \begin{Thm}[\bf{Probability of error}]\label{thm:error}
        The probability of choosing a wrong hypothesis $\theta \notin \Theta_k^\star$ for agent $k\in\mathcal N$ is upper bounded by:
        \begin{align}
            \mathbb P \left\{\theta \in \widehat {\boldsymbol\Theta}_k\right\} \leq &\; \frac 4M \textup{Tr} \left(\R_{\bL}\right) \sum_{\theta^\star \in \Theta_k^\star} D^{-1}_{\textup{KL}}\big(L_k\left(\theta^\star\right) \mid\mid L_k\left(\theta\right)\big) \nonumber\\
            &+ O(\mu/\delta^2) + O\left( 1 / \delta^5 M^2 \right)
        \end{align}
    \end{Thm}
    \begin{proof}
    First, we upper bound the probability using the definition of $d(\cdot,\cdot)$ and its estimate from~(\ref{eq:Ldist_estimation}) and (\ref{eq:Ldist_estimation2}), along with the properties of probability. 
    For any $\theta_j \notin \Theta_k^\star$, we have that:
    \begin{align}
        \mathbb P \left\{\theta_j \in \widehat {\boldsymbol{\Theta}}_k\right\} &\;\leq \mathbb P \left\{\exists \theta_k^\star \in \Theta_k^\star \colon \widehat {\boldsymbol d }_k(\theta_k^\star, \theta_j) < 0 \right\} \nonumber\\
        &\;\leq \sum_{\theta_k^\star \in \Theta_k^\star}\mathbb P \left\{\widehat {\boldsymbol d }_k(\theta_k^\star, \theta_j) < 0 \right\}
        \label{eq:th_h1}
    \end{align}
    Next, we estimate the probability of $\widehat{\boldsymbol d }_k(\theta_k^\star, \theta_j)$ being negative for some fixed $\theta_j$ and $\theta_k^\star$ using~(\ref{eq:Ldist_estimation2}), while denoting $j^\star_k$ as the index of $\theta_k^\star$:
    \begin{align}
        &\mathbb P\left\{\widehat{\boldsymbol d }_k(\theta_k^\star, \theta_j) < 0\right\} = \mathbb P \left\{[\widehat \bL]_{k,j} - [\widehat \bL]_{k, j^\star_k} < 0\right\} \nonumber\\
        &= 1 - \; \mathbb P \Big\{[\widehat \bL]_{k,j^\star_k} - [\overline \cL]_{k,j^\star_k} - \left([\widehat \bL]_{k, j} - [\overline \cL]_{k, j}\right) \nonumber\\
        &\;\;\;\;\;\;\;\;\;\;\;\;\;\;\;\;\;\;\leq [\overline\cL]_{k,j} - [\overline\cL]_{k,j^\star_k}\Big\} \nonumber\\
        &\leq 1 - \mathbb P \Big\{\left|[\widehat \bL]_{k, j^\star_k} - [\overline \cL]_{k, j^\star_k}\right| + \left|[\widehat \bL]_{k, j} - [\overline \cL]_{k, j}\right|\nonumber\\
        &\;\;\;\;\;\;\;\;\;\;\;\;\;\;\;\;\;\; \leq [\overline\cL]_{k,j} - [\overline\cL]_{k, j^\star_k}\Big\} \nonumber\\
        &\leq  1 - \mathbb P \left\{\left|[\widehat \bL]_{k,j^\star_k} - [\overline \cL]_{k, j^\star_k}\right| \leq \left([\overline\cL]_{k,j} - [\overline\cL]_{k,j^\star_k}\right)/2\right\} \nonumber\\
        &\;\;\;\;\;\; \times \mathbb P \left\{\left|[\widehat \bL]_{k, j} - [\overline \cL]_{k, j}\right| \leq\left([\overline\cL]_{k,j} - [\overline\cL]_{k, j^\star_k}\right)/2\right\}
        \label{eq:th_h0}
    \end{align}
    We can transform the result~(\ref{eq:prev_result}) from~\cite[Theorem 2]{shumovskaia2023discovering} into:
    \begin{align}
        \bE \left|[\widehat\bL]_{k,j} - [\overline\cL]_{k,j} \right| \leq \frac 1M \textup{Tr} \left(\R_{\bL}\right) + O(\mu/\delta^2) + O\left( 1 / \delta^5 M^2 \right)
        \label{eq:mark}
    \end{align}
    By Markov's inequality~\cite{sayed_2023}, for any $a > 0$:
    \begin{align}
        &\mathbb P \left(\left|[\widehat\bL]_{k,j} - [\overline\cL]_{k,j} \right| \leq a\right) \nonumber\\
        &\geq  1 - \frac 1{aM} \textup{Tr} \left(\R_{\bL}\right) + O(\mu/\delta^2) + O\left( 1 / \delta^5 M^2 \right)
    \end{align}
    Also, by the definition of KL divergence we have that:
    \begin{align}
        [\overline\cL]_{k,j} - [\overline\cL]_{k, j^\star_k} =&\; D_{\textup{KL}}\left(L_k\left(\theta_k^\star\right) \mid\mid L_k\left(\theta_j\right)\right) > 0.
    \end{align}
    Thus,~(\ref{eq:th_h0}) can be upper bounded by:
    \begin{align}
        &\mathbb P\left\{\widehat{\boldsymbol d }_k(\theta_k^\star, \theta_j) < 0\right\} \nonumber\\
        &\leq 1 - \bigg(1 - \frac{\frac 2M \textup{Tr} \left(\R_{\bL}\right) + O(\mu/\delta^2) + O\left( 1 / \delta^5 M^2 \right)}{[\overline\cL]_{k,j} - [\overline\cL]_{k, j^\star_k}}\bigg)^2 \nonumber\\
        & \approx 4 M^{-1} \textup{Tr} \left(\R_{\bL}\right) D^{-1}_{\textup{KL}}\big(L_k\left(\theta_k^\star\right) \mid\mid L_k\left(\theta\right)\big) \nonumber\\
        &\;\;\;\;+ O(\mu/\delta^2) + O\left( 1 / \delta^5 M^2 \right)
        \label{eq:tay}
    \end{align}
    using the Taylor's expansion for any small $x$, namely, $(1+x)^2 = 1 + 2x + O(x^2)$.

    Combining~(\ref{eq:th_h1}) with~(\ref{eq:tay}) we get the desired statement.
    \end{proof}

    {\color{black}The model's performance is influenced by parameters $\delta$ and $\mu$, with $\mu$ being arbitrarily small. As shown in~\cite{shumovskaia2023community}, when $\delta$ is close to 1, agents rely more on their own observations, making it easier to reveal their true state in social learning. This aligns with the derived result.}

\section{Computer Experiments}

    In this section, we consider the image dataset MIRO (Multi-view Images of Rotated Objects)~\cite{kanezaki2018rotationnet}, which contains objects of different classes from different points of view -- see Fig.~\ref{fig:miro_example}. 
    For each class, there are 10 objects, and each of the objects has 160 different perspectives. 

    \begin{figure}
        \centering
        \includegraphics[width=\linewidth]{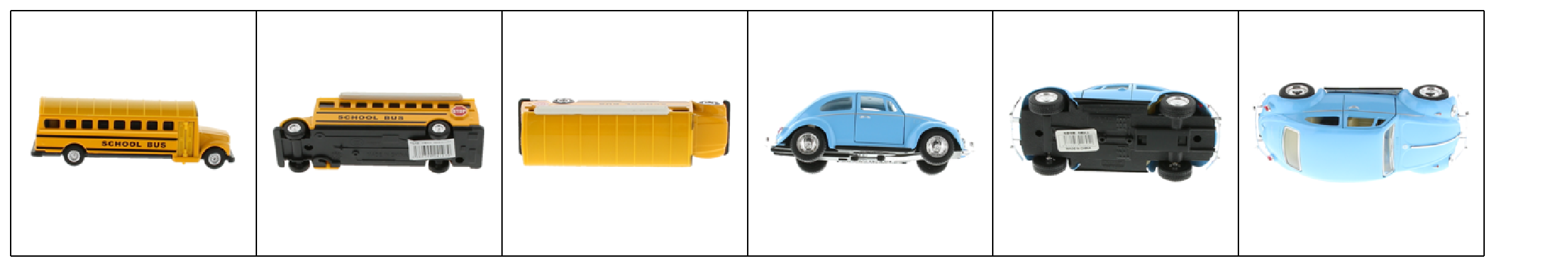}
        \caption{Example of images from the MIRO dataset for classes ``bus" and ``car".}
        \label{fig:miro_example}
    \end{figure}
    
    A network of agents wishes to solve a binary hypotheses problem to distinguish between states $\theta_0$ corresponding to the class ``bus" and $\theta_1$ corresponding to the class ``car". 
    Each agent has its own convolutional neural network (CNN) classifier. 
    These CNNs are trained to distinguish classes $\theta_0$ and $\theta_1$ by observing only a part of the image, similar to the approach in~\cite{hu2023non,bordignon2022partial}. Each image measures $224 \times 224$ pixels, and each agent observes a section of size $112 \times 112$ pixels, situated in different regions of the image. We illustrate the observation map in Fig.~\ref{fig:miro_train}. 
    The CNN architecture consists of three convolutional layers: 6 output channels, $3 \times 3$ kernel, followed by ReLU and $2 \times 2$ max pooling; 16 channels, $3 \times 3$ kernel, ReLU, and $2 \times 2$ max pooling; 32 channels, $3 \times 3$ kernel, ReLU, and $2 \times 1$ max pooling. This is followed by linear layers of sizes $288 \times 64$, $64 \times 32$, and $32 \times 2$, with ReLU activation function in between. The final prediction layer is log softmax. Training involves 100 epochs with a learning rate of $0.0001$ and negative log-likelihood loss.

    For generating a combination matrix (see Fig.~\ref{fig:miro_train}), we initially sample an adjacency matrix following the Erdos-Renyi model with a connection probability of $0.2$. Subsequently, we set the combination weights using the averaging rule~\cite[Chapter 14]{Sayed_2014}. During the inference, we let the central agent be malicious -- see Fig.~\ref{fig:miro_test}.

    Since we only have 10 objects of each class, having only a handful of objects as a test subset is not enough to provide a reliable accuracy metric. 
    Thus, we perform a cross-validation procedure where at first, we train the CNNs on 9 objects from each class, leaving 1 object from each class for testing purposes.
    On average, the cross-validation accuracy of standalone classifiers is \textbf{0.68}. 
    The value is relatively low due to a small training set and limited observation available at each agent.
    Given that many folds had some classifiers with an accuracy below 0.5, we decided to retain only those folds where each agent achieved at least 0.5 accuracy.
    As a result, we are left with 72 folds instead of 100 with the mean accuracy of standalone classifiers equal to \textbf{0.81}.
    
    We apply the adaptive social learning strategy with $\delta=0.1$ over 480 iterations, showing each frame 3 times on average. The network observes a ``bus" while the central agents observe a ``car" (Fig.\ref{fig:miro_test}).
    We can see that despite the presence of the malicious agent, the average belief of each agent tends towards the correct hypothesis $\theta_0$ (see Fig.~\ref{fig:miro_asl}) with the mean accuracy {\textbf{0.8}}. 
    However, as depicted in Fig.~\ref{fig:miro_accuracy}, the algorithm is able to identify the malicious agent achieving the mean accuracy {\textbf{0.99}}.

    \begin{figure}
        \centering
        \begin{subfigure}[b]{0.22\textwidth}
            \centering
            \caption{Training scheme.}
            \includegraphics[width=\linewidth]{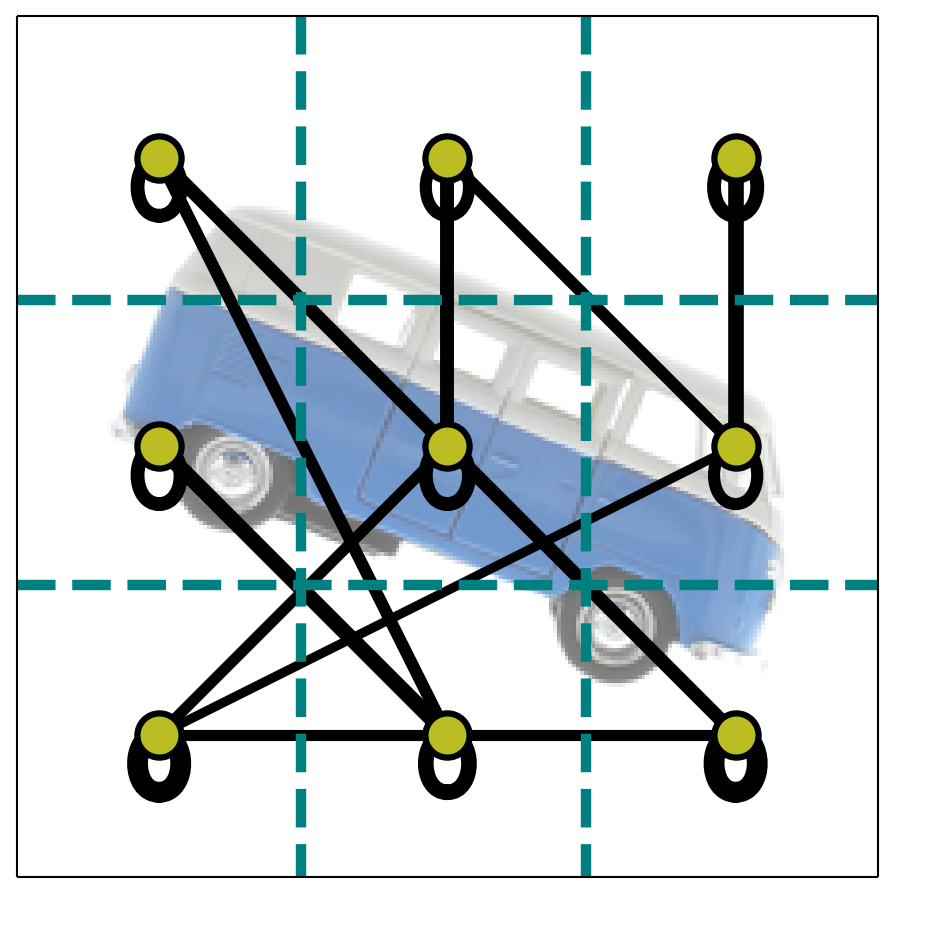}
            \label{fig:miro_train}
        \end{subfigure}
        \hfill
        \begin{subfigure}[b]{0.22\textwidth}
            \centering
            \caption{Test scheme with the central node being malicious.}

            \includegraphics[width=\linewidth]{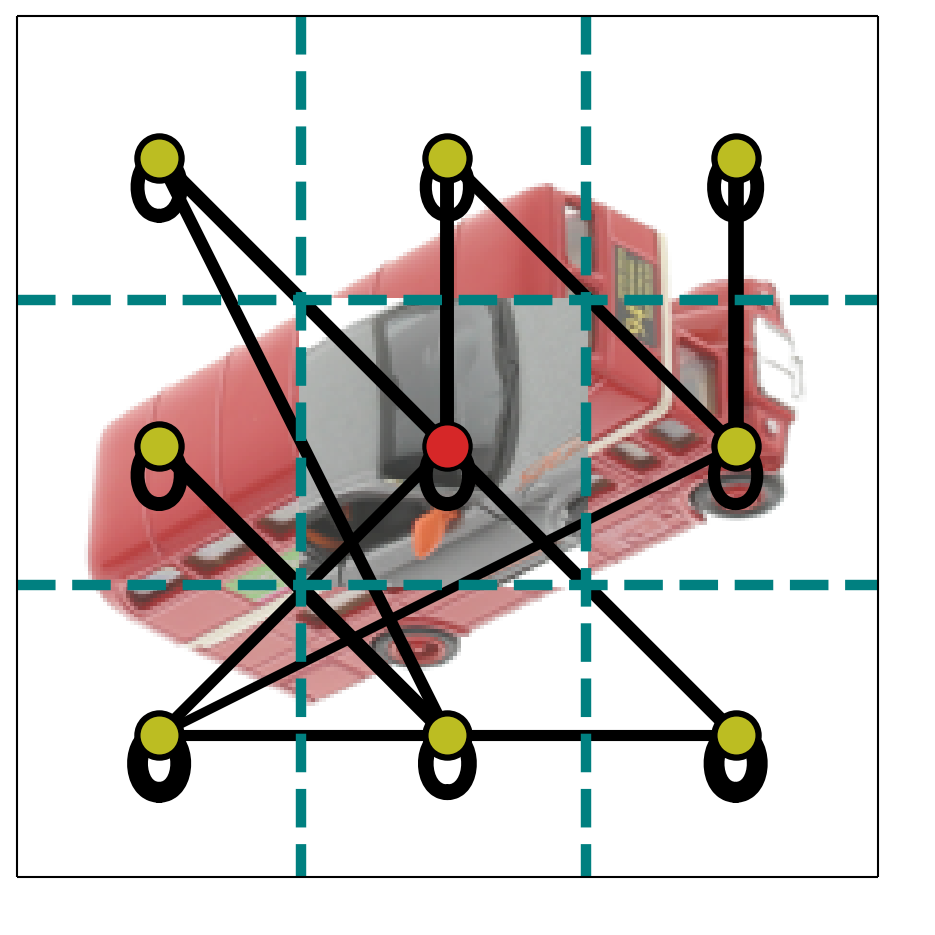}
            \label{fig:miro_test}
        \end{subfigure}
        
        \caption{Observation map of each agent.}
        \label{fig:graph_train}
    \end{figure}

    \begin{figure}
        \centering
        \begin{subfigure}[b]{0.22\textwidth}
            \centering
            \caption{Accuracy of the social learning strategy to predict $\theta_0$.}
            \includegraphics[width=\linewidth]{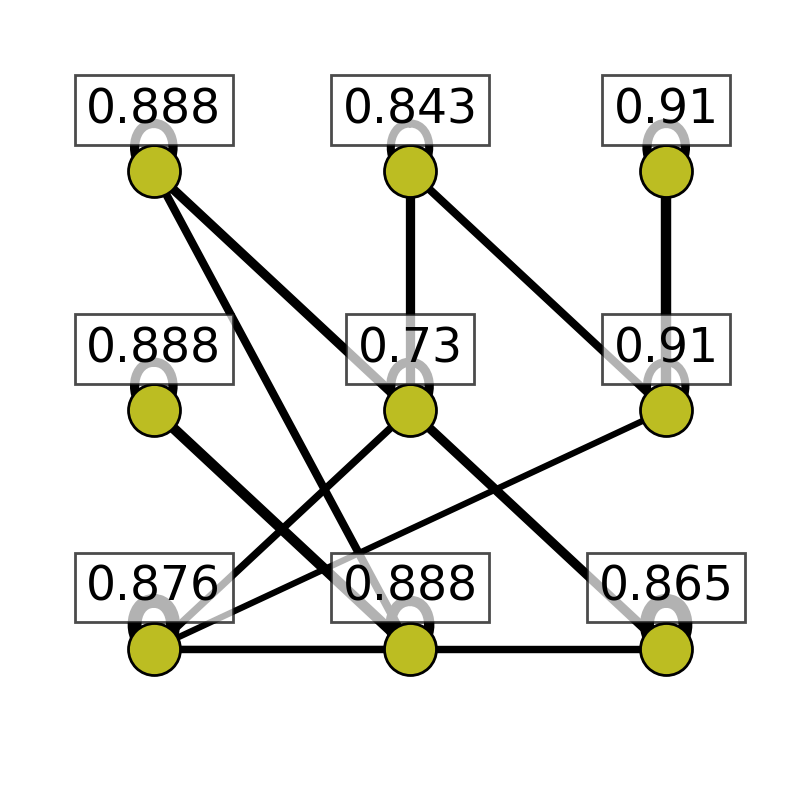}
            \label{fig:miro_asl}
        \end{subfigure}
        \hfill
        \begin{subfigure}[b]{0.22\textwidth}
            \centering
            \caption{Malicious detection accuracy and learned graph.}
            \includegraphics[width=\linewidth]{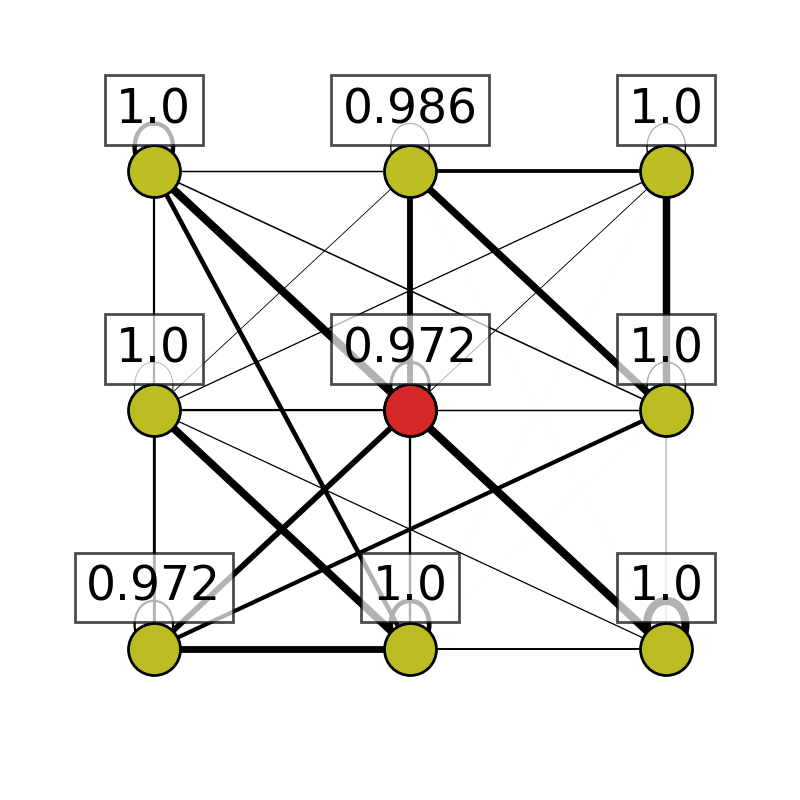}
            \label{fig:miro_accuracy}
        \end{subfigure}
        \caption{Accuracy of the adaptive social learning strategy~\cite{bordignon2020adaptive} and Algorithm~\ref{alg}. {\color{black}Yellow represents $\theta_0$, and red represents $\theta_1$.} For each fold, social learning accuracy is averaged over the past 100 iterations.}
        \label{fig:miro_acc}
    \end{figure}

 \newpage 
\bibliographystyle{IEEEtran}
\bibliography{references}

\end{document}